\documentclass[journal]{IEEEtran}
\usepackage{cite}

\ifCLASSINFOpdf
   \usepackage[pdftex]{graphicx}
   \usepackage{tikz}
\else
\fi
\usepackage{amsmath}
\usepackage{amsthm}
\newtheorem{corollary}{Corollary}[section]

\newtheorem{theorem}{Theorem}
\usepackage{algorithmic}
\begin{document}
%
\title{Boolean Expressions in Firewall Analysis}
%
%
%

\author{Adam~Hamilton,~\IEEEmembership{University of Adelaide,}
        Matthew~Roughan,~\IEEEmembership{University of Adelaide,}
        and~Giang~Nguyen,~\IEEEmembership{University of Adelaide}
\thanks{}
\thanks{Adam Hamlton, Giang Nguyen and Matthew Roughan are associated with The Faculty of Sciences Engineering and technology, University of Adelaide, South Australia, Australia email: adam.h.hamilton@adelaide.edu.au}
\thanks{Manuscript received April 19, 2005; revised August 26, 2015.}}

%
%

\markboth{IEEE Transactions on Network and Service Management}%
{Shell \MakeLowercase{\textit{et al.}}: Bare Demo of IEEEtran.cls for IEEE Journals}
%



\maketitle

\begin{abstract}
Firewall policies are an important line of defence in cybersecurity, specifying which packets are allowed to pass through a network and which are not. These firewall policies are made up of a list of interacting rules. In practice, firewall can consist of hundreds or thousands of rules. This can be very difficult for a human to correctly configure \cite{dineshathesisalma9927943178101811, wool1306389}. One proposed solution is to model firewall policies as Boolean expressions and use existing computer programs such as SAT solvers to verify that the firewall satisfies certain conditions. This paper takes an in-depth look at the Boolean expressions that represent firewall policies. We present an algorithm that translates a list of firewall rules into a Boolean expression in conjunctive normal form (CNF) or disjunctive normal form (DNF). We also place an upper bound on the size of the CNF and DNF that is polynomial in the number of rules in the firewall policy. This shows that past results suggesting a combinatorial explosion when converting from a Boolean expression in CNF to one in DNF does note occur in the context of firewall analysis.
\end{abstract}


%
\IEEEpeerreviewmaketitle

\section{Introduction}
%
%
%
%
\IEEEPARstart{F}{irewalls} form one of the first lines of defence in cybersecurity. They help to protect the network by implementing a security policy, a set of rules that tells the firewall which packets are allowed to pass through the firewall and which packets cannot. 

The main motivation behind this research is that firewalls are complicated \cite{wool1306389}. When there are a sufficiently large number of interconnected rules, it becomes difficult for a human to keep track of the policy enforced by the firewall \cite{wool1306389, dineshathesisalma9927943178101811}. Errors in firewall configurations are common in practice, and  can lead to large security holes that compromise entire networks \cite{wool1306389}. To make things easier for a human operator, it is important to be able to reason about the firewall's security policy in a computational and mathematical manner.

A useful tool in the analysis of firewall policies is Boolean algebra. In a large number of scenarios, the firewall's security policy can be viewed as a Boolean expression. Each rule in the security policy consists of an action and a predicate of a particular type. The firewall works by taking a packet and comparing the packet's data against these rules. 

Previous work \cite{Heule2016AnalysisOC, HazelhurstS2000Afit, HazelhurstScott2000AfAF, QBF} has shown this link between Boolean expressions and firewall policies, but relatively little has gone into the specific structure of the Boolean expression. We hope this paper can bridge this gap in exploring the Boolean structure of firewall policies. 

The main contributions of our paper are:

\begin{enumerate}
    \item We provide an efficient encoding of stateless firewall policies as Boolean expressions with bounds on the size of the resultant Boolean expressions. 
    \item We give an algorithm for translating in polynomial time, firewall policies from Boolean expressions in CNF to Boolean expressions in DNF and back. 
    \item We prove that the combinatorial explosion when translating between Boolean expressions in CNF to DNF does not occur in firewall analysis as previously suggested in \cite{latticesalma9926252301811}.
\end{enumerate}

\section{Background and related work}
\subsection{Firewall analysis}
A firewall is a method of filtering traffic on a computer network. A firewall sits either at the boundary of, or within a computer network and either allows packets to pass through the firewall or drops the packet from the network. The firewall policies that we consider consist of an ordered list of rules; each is made up of $d$ intervals followed by an action which is either $accept$ or $reject$. For instance, a firewall rule could say that if the destination port is between 20 and 40, then allow the packet to pass through the network. A packet is compared to a rule in the firewall by comparing the  number in the packet's header to the corresponding interval specified by the rule. If the numbers in the packet's header belong to each of the corresponding intervals, then the rule's action is applied to the packet. 

Here we use the term \textit{firewall policy} to denote a Boolean expression that represents the set of packets that are allowed through a firewall. A \textit{firewall rule list} is a method of expressing a firewall using an ordered list of rules.

Most firewall rules focus on just five packet header fields: protocol number, source and destination IP addresses, and source and destination port numbers \cite{LiuAlexX2010Fdaa}. There are other types of firewalls that employ more complex methods of packet filtering such as stateful policies, or deep packet inspection, but these are outside the scope of the paper. We do allow for filtering on more than just the five-tuple described but the important detail is that there can only be a bounded number of such fields without significantly altering the TCP/IP protocol. 

\subsection{Boolean algebra}

When we restrict firewalls to only allow or deny incoming and outgoing packets, a firewall policy is equivalent to a Boolean expression specifying which packets will be permitted through the firewall. This idea is useful because it allows us to express problems in firewall analysis as Boolean satisfiability (SAT) problems, and to use tools such as SAT solvers to extract information about the policy \cite{Heule2016AnalysisOC, HazelhurstS2000Afit, HazelhurstScott2000AfAF, QBF}. In this section, we define some common terminology used in Boolean algebra. 

A Boolean expression consists of a set of variables $x_i$ which can take a value of 0 or 1. A Boolean expression $EXP$ consists of either a variable, the negation of an expression $\neg EXP$, the disjunction of two expressions $EXP_1 \vee EXP_2$, or the conjunction of two expressions $EXP_1 \wedge EXP_2$. For instance, the following is a Boolean expression: $(x_1\vee \neg x_2) \vee \neg(x_3 \wedge \neg x_1).$

A \textit{normal form} is a method of specifying a mathematical object that has a particular syntactic form. A Boolean expression is in conjunctive normal form (CNF) if it is the conjunction of one or more clauses: $C_1\wedge C_2\wedge \dots \wedge C_m$, where each clause $C_i$ is of the form $v_1\vee v_2 \vee \dots \vee v_{n_i}$ and each $v_i$ is either a variable or the negation of a variable. Similarly, a Boolean expression is in disjunctive normal form (DNF) if it is the disjunction of numerous clauses: $C_1\vee C_2\vee \dots \vee C_m$, where each clause $C_i$ is of the form: $v_1\wedge v_2 \wedge \dots \wedge v_{n_i}$ and each $v_i$ is either a variable or the negation of a variable. 
The restrictive syntax of normal forms allows one to make statements about the computational complexity of certain problems. For instance, the problem of determining if a Boolean expression can evaluate to 1 is called Boolean satisfiability (SAT). If a Boolean expression is in CNF and every clause consists of at least 3 variables then solving SAT is an NP-complete \cite{COOK10.1145/800157.805047} problem. However, solving SAT on a Boolean expression in DNF can be solved in linear time in the size of the DNF expression \cite{latticesalma9926252301811}.

\subsection{Boolean algebra in firewall analysis}
Hazelhurst \cite{HazelhurstScott2000AfAF, HazelhurstS2000Afit} uses Boolean models to make firewall analysis simpler. The author argues that by expressing firewall policies as Boolean decision diagrams, managing firewalls becomes easier for a human to understand. Elmallah and Gouda \cite{AnalysisComplexity} use Boolean algebra to place lower bounds on the complexity of certain problems in firewall analysis. The authors define twelve decision problems in firewall analysis, and show how they can be expressed as Boolean SAT problems. 

Like Boolean expressions, firewall policies can also be expressed in normal form. In this paper we consider whitelists and blacklists. These are normal forms of firewalls that are very similar to CNF and DNF of Boolean expressions. A whitelist is a firewall policy that consists entirely of rules of the form $<predicate> \rightarrow allow$ followed by a default rule that denies all other packets. Likewise, a blacklist is a firewall policy whose rules are of the form $<predicate> \rightarrow deny$ followed by a default rule that accepts all remaining packets. Whitelists and blacklists have been studied by Ranathunga \textit{et al.} \cite{ranathunga2016mathematical}. There are advantages to expressing firewalls as blacklists and whitelists. One notable advantage is that rules commute when expressed in this way. This in turn has advantages when it comes to reasoning about firewall configurations. When one wants to add an extra rule to a firewall in whitelist form, it is not necessary to work out where in the rule list it needs to placed. The same also applies for firewalls expressed as blacklists.

In \cite{ranathunga2016mathematical} the authors also discuss the advantages of whitelists from a security perspective. In a safety critical computer system, such as a power plant, where malicious behaviour can lead to large amounts of damage and loss of life it is far better to deny a trusted user than allow an untrustworthy user access. Ranathunga \textit{et al.} as well as several published best-practice guides to firewall configuration \cite{dineshathesisalma9927943178101811} argue that in a whitelist it is harder for one to accidentally allow untrustworthy packets past the firewall, because one is forced to think about the packets that are allowed through. on the other hand, when one uses deny rules, it can be easier to forget to specify some undesirable packets. 

Firewall policies expressed as a whitelist are similar to Boolean expressions in DNF, in the sense that a packet is able to pass through the firewall if its header is contained in the disjunction of the firewalls rules; that is, if the packet satisfies at least one of the rules. Likewise, an input satisfies a Boolean expression in DNF if it satisfies at least one of the clauses. Similarly, a blacklist firewall policy is similar to a Boolean expression in CNF in that a packet is allowed through the firewall it does not satisfy any of the reject rules. More formally, a packet is allowed through a firewall blacklist if it is contained in the conjunction of the negation of each rule. We use this similarity later in the paper to define a method of translating firewall policies into Boolean expressions in normal forms.  

\subsection{Assumptions in this paper}
In their paper \cite{AnalysisComplexity}, Elmallah and Gouda define 12 problems in firewall analysis and show that they are NP-complete. In doing so, they show that a Boolean 3-SAT problem can be translated into a firewall policy in such a way that there is a one-to-one correspondence between rules in the firewall and clauses in the 3-SAT expression. This mapping from 3-SAT to firewalls can be done in polynomial time, proving the NP-completeness of Elmallah and Gouda’s 12 problems in firewall analysis. However, one of the key ingredients in their proof is introducing an arbitrary number of packet header fields. We do not believe this to be a realistic depiction of a computer network. The TCP/IP protocol suite consists of packets that have a finite number of header fields \cite{NetworkDetails}. In most works on firewall management there are only 5 packet header fields that are considered. Even in more complex situations we can assume an upper bound of 10 to 12 fields. 

In order to analyse computational complexity, it is important to specify which variables are being held constant and which are allowed to increase. In this paper, we keep the number of packet header fields constant and allow the number of possible field values such as addresses to increase. This number of possible addresses is called $B$. It is assumed that $B$ is large but $\log(B)$ is small or at least manageable. For instance, in IPV6, $B$ would be $2^{128}$. 

Heule et al. \cite{HeuleEtAl} reinforce the idea  that the number of fields in a packet header can almost always be assumed to be constant. In \cite{HeuleEtAl}, the authors describe a polynomial-time algorithm to answer the firewall implication problem in \cite{AnalysisComplexity}. The firewall implication problem is ``given two firewalls $F_1$ and $F_2$, is it the case that all the packets accepted by $F_1$ are also accepted by $F_2$?" We extend this work to provide a framework for designing polynomial time algorithms for all the problems in \cite{AnalysisComplexity}.

Another assumption we make is that every predicate in each rule of the firewall is made up of queries of the form $X>Y$ or $X = Y$, that is, we treat each packet header field as a number and we focus only on queries about the size of the number. Our complexity results would not hold if the firewall was capable of querying individual bits of a packet header.

\subsection{Firewall decision trees}
\label{Firewall decision trees}
Firewall decision trees are a useful way of representing a firewall policy. We use this model of firewall policies as an intermediate step between firewall policies given as rule lists, and firewall policies given as Boolean expressions. The reason for using firewall decision trees as an intermediate model is so we can use results from \cite{Acharya2010FirewallMA, RulesInPlay}, which place bounds on the size of the firewall decision tree and use these results to place bounds on the size of the Boolean expression. 
A firewall decision tree is a weakly-connected acyclic digraph used to represent firewall policies. A node in the digraph with no outgoing edges is called \textit{terminal}. Every terminal node in the decision tree is assigned an action, in our case, this action is $accept$ or $deny$. Every non-terminal node is marked with a name of a field and has outgoing edges marked with an interval of values. The values on different outgoing links from the same non-terminal node do not overlap. There is a single node called the \textit{root node} such that all edges connected to this node are outgoing edges. Given any packet $p$, there is exactly one path from the root that finishes at a terminal node. This terminal node defines the decision made by the firewall for the packet $p$. For an example of a firewall decision tree, see Figure \ref{completetree}.

We define the \textit{depth} of a node $N$ in a tree as the number of edges in the path from the root node to $N$. Because the decision tree is a tree, this path is unique. For simplicity, we assume that all nodes of a given depth are marked with the same field name and that all terminal nodes are located at the same depth. This convention is not strictly necessary; we could omit certain non-terminal nodes. This would correspond to a firewall in which some rules do not depend on certain packet fields. We could also permute some non-terminal nodes within the tree. This would correspond to a rule where you check the packet fields in a different order. We ignore these for simplicity. 

A decision tree is a simple deterministic finite automaton. To evaluate a packet against a firewall decision tree, we follow a path from the root down to the leaf node choosing at each node, the edge whose label corresponds to the packet's field value. We can do this by performing a binary search at each node, looking for the interval that corresponds to our given packet. Since the intervals in the outgoing edges of each node are disjoint, this path from root node to terminal node is unique. If we order the outgoing edges of each node, then we can match a packet against a firewall in $O(\log(n)d)$ time, where $n$ is the number of rules and $d$ is the number of fields in a packet header. 

\section{Algorithm for generating firewall decision trees.}

In this section, we present an algorithm for translating a firewall rule list into a firewall decision tree. Based on this algorithm, we can place bounds on the size of the tree and later place bounds on the efficiency of CNF and DNF translation. 

Let $I$ and $J$ be disjoint intervals on the natural numbers such that the smallest element of $J$ is larger than the largest element of $I$. We define the interval $b(I,J)$ as the set of all natural numbers between $I$ and $J$, $b(0,J)$ the interval of natural numbers smaller than any element of $J$, and $b(I,\infty)$, the set of natural numbers larger than any element of $I$.

We use $T$ to denote a firewall decision tree. Let $I_i$ be the intervals in the outgoing edges from the origin of $T$, and $T_i$ be the subtrees whose root is the end of $I_i$. We abuse notation and use $I_i$ to denote both the directed edge and its label.

We now recursively define the $addrule$ algorithm. It is a recursive algorithm that takes a rule $r$ and an existing policy, finds which parts of the rule are already covered by the policy, and then splits $r$ into sub-rules that are not covered by the policy. A rule is defined as a list $(r_1, r_2, \dots, r_n, A)$, where each $r_i$ denotes an interval over the natural numbers and $A$ is an action, either accept or deny. We use $r(2, :)$ to denote rule $r$ with the first entry removed, similar to the command ``cdr" in the LISP programming language. A rule can consist of no intervals and a single action $A$. We begin with the rule $r$ and tree $T$ as inputs. 

\begin{figure}
\centering
    \textbf{The Addrule algorithm}\par\medskip
\begin{algorithmic}
\STATE addrule(r, T)
\IF {$r$ and $T$ both consist of an action} 
  \STATE $addrule(r,T) = T$
\ELSE
  \STATE let $I_r$ be the first interval in the rule $r$.
  \FOR {each interval $I_i$ in the outgoing edges of the root node of $T$}
    \IF{$I_i \cap I_r$ is nonempty}
      \STATE split $I_i$ into two (possibly empty) edges $I_i \cap I_r$ and $I_i \setminus I_r$. A the end of $I_i \setminus I_r$ we place the subtree $T_i$, and at the end of $I_i \cap I_r$, we place the subtree $addrule(r(2,:), T_i)$ (the rule $r$ could still refer to packets not in the domain of the tree).
    \ELSE
      \STATE keep $I_i$ and its subtree $T_i$ in $T$ 
    \ENDIF
  \ENDFOR
  \STATE consider the intervals $b(I_i, I_{i+1})$, (the intervals between the outgoing edges of the origin of $T$)
  \FOR {each interval in [$b(I_i, I_{i+1})$]}
    \IF{$b(I_i, I_{i+1})\cap I_r$ is nonempty}
      \STATE create an edge from the origin of $T$, with the label $b(I_i, I_{i+1})\cap I_r$ and the subtree $r(2,:)$ at the end of it. 
    \ENDIF
  \ENDFOR
\ENDIF 
\end{algorithmic}
\caption{An algorithm that takes a firewall decision tree and a rule and returns a new decision tree.}

\end{figure}

We use vectors and trees interchangeably when we talk about the tree $r(2,:)$. We represent a vector $v = ( v_1, v_2,\dots, v_m )$ as a tree $T_v$, by defining a tree with $m+1$ nodes $n_0, n_1, \dots , n_m$. We put directed edges between the nodes $n_i, n_{i+1}$ for $i = 0,\dots n - 1$, and the edge between $n_{i-1}, n_i$ is marked with $v_i$. Essentially, when we express a vector as a tree, we draw a sequence of connected lines each marked with the vectors entries. 

In order to give the reader a more intuitive understanding of the notation used in this paper, we have provided a diagram, Figure \ref{fig:intervals}, to demonstrate this notation. 

\begin{figure}[h]
    \centering
    \includegraphics[width=0.5\textwidth]{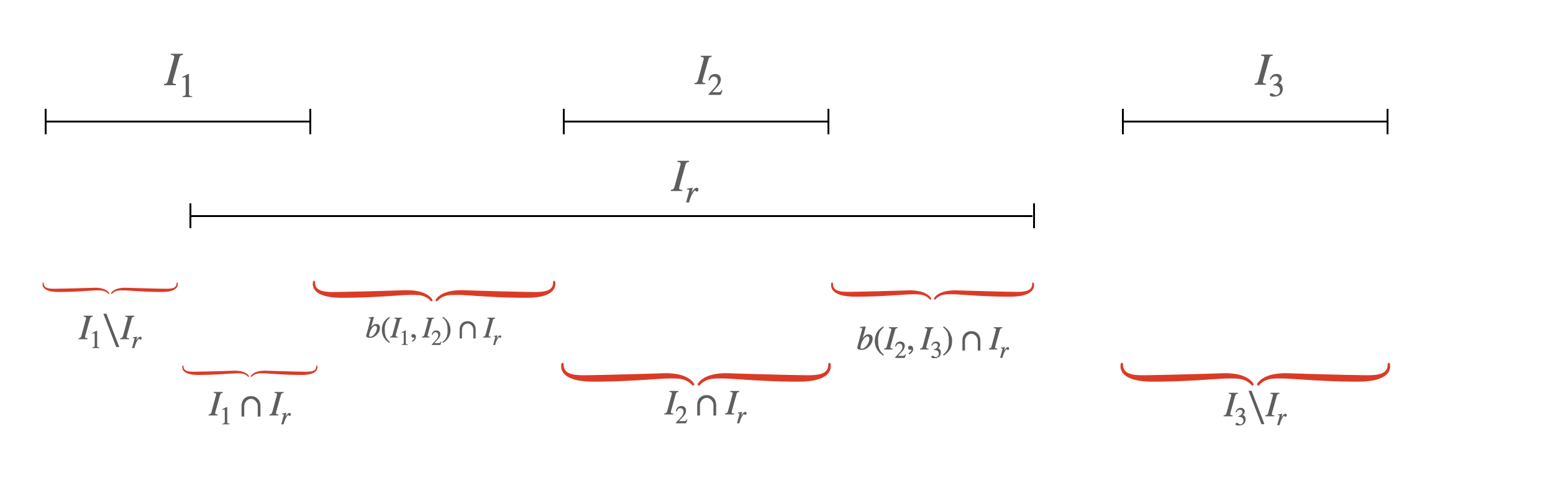}
    \caption{The intervals that we consider in the first iteration of $addrule(r,T)$. The intervals $I_r$, represents the first interval of rule $r$, and $I_1$, $I_2$, and $I_3$ are the intervals from the root node of $T$. We use the red brackets to denote the intervals that we need to consider in the algorithm. We have made a point of only showing the non-empty intervals. For instance, since $I_3$ is completely disjoint from $I_r$, we do not consider the interval $I_3\cap I_r$, likewise, we do not need to consider the empty interval $I_3\setminus I_r$, but we need to consider both $I_1\cap I_r$ and $I_1\setminus I_r$, since $I_1$ and $I_r$ overlap.}
    \label{fig:intervals}
\end{figure}

The diagrams in Figures \ref{fig:tree1} and \ref{fig:intervals2}, show one iteration of the $addrule$ algorithm, showing the new branches created in the firewall decision tree, as well as their labels and new subtrees. 

\begin{figure}[h]
    \centering
    \includegraphics[width=0.35\textwidth]{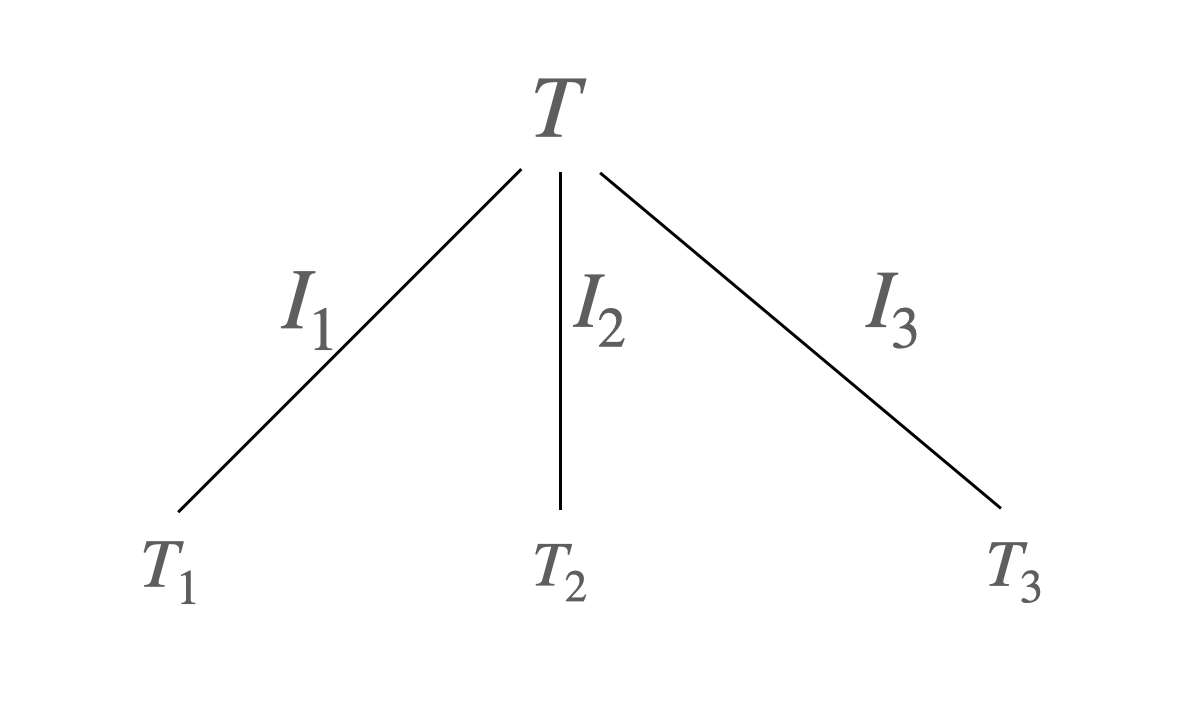}
    \caption{An abstract example of a firewall policy tree $T$ corresponding to the rules in Figure \ref{fig:intervals} to which we will add a rule $r$. The directed edges from the root node of $T$ are labelled with intervals $I_1$, $I_2$, and $I_3$, where $I_1<I_2<I_3$. The nodes at the ends of $I_1$, $I_2$, and $I_3$ are labelled with $T_1$, $T_2$, and $T_3$, where $T_i$, represents both the node at the end of $I_i$ and the subtree who's root node is $T_i$.}
    \label{fig:tree1}
\end{figure}

\begin{figure}[h]
    \centering
    \includegraphics[width=0.5\textwidth]{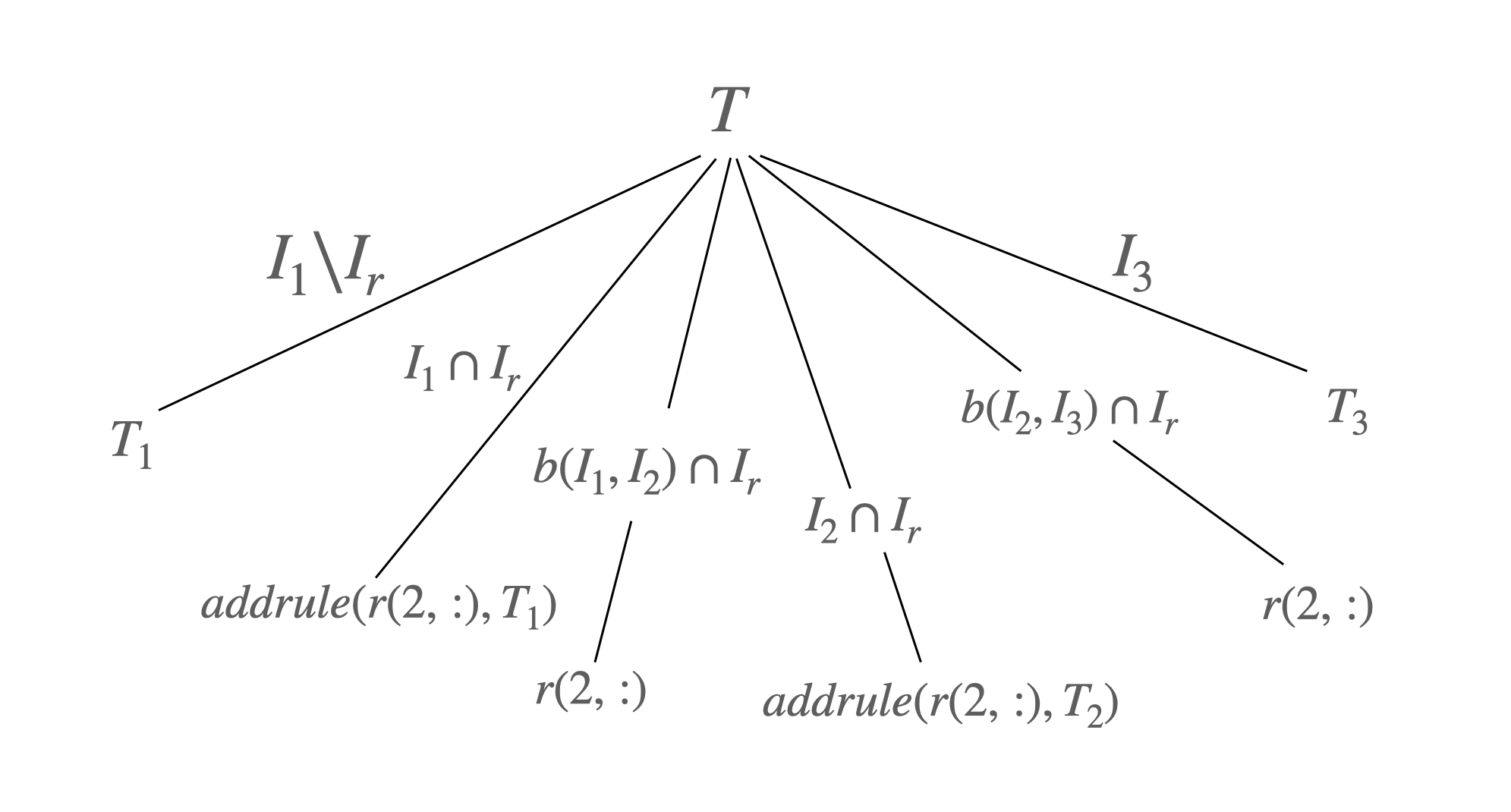}
    \caption{The first layer of the firewall decision tree shown in Figures \ref{fig:intervals}-\ref{fig:tree1} after adding $r$ from Figure \ref{fig:intervals}. The three edges from the root node are replaced by the five intervals from Figure \ref{fig:intervals}. The subtrees at the end of each edge is $T_i$ if the interval contains part of $I_i$ and no part of $I_r$ are given by $T_i$. The subtrees at the end of an interval that contains both $I_r$ and $I_i$ are given by $addrule(r(2,:), T_i)$, and the subtrees at the end of an interval that just consists of $I_r$ are given by $r(2,:)$}
    \label{fig:intervals2}
\end{figure}


\subsubsection{Concrete example of the addrule algorithm}

Here we show an example of the algorithm being used to construct the decision tree for the following rules in interval notation:
$$
(1, 10)(2, 5)(1, 10) \rightarrow \text{accept}
$$
$$
(3, 15)(3, 4)(1, 10) \rightarrow \text{deny}.
$$

\begin{figure}[h!]

\begin{center} 
\begin{tikzpicture}[scale=.7,colorstyle/.style={circle, draw=black!100,fill=black!100, thick, inner sep=0pt, minimum size=2 mm}]
    \node at (-5,-1)[colorstyle,label=left:$\text{origin}$]{};
    \node at (-5,-3)[colorstyle]{};
    \node at (-5,-5)[colorstyle]{};
    \node at (-5,-7)[colorstyle]{};
    \draw [thick](-5, -1) -- (-5, -3) node [midway, right, fill=white] {(1, 10)};
    \draw [thick](-5, -3) -- (-5, -5) node [midway, right, fill=white] {(2, 5)};
    \draw [thick](-5, -5) -- (-5, -7) node [midway, right, fill=white] {(1, 10)};
    
\end{tikzpicture}
\caption{The decision tree created from a single rule. This takes the list of intervals present in rule $r = (1, 10)(2, 5)(1, 10) \rightarrow \text{accept}$ and creates a tree with 4 nodes, and the edges are labelled with the intervals in the rule $r$. }
\label{fig:step1}
\end{center}

\end{figure}
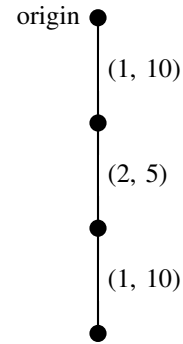
We begin by creating a decision tree from a single rule. This is shown in Figure \ref{fig:step1}.

We then take the second rule $(3, 15)(3, 4)(1, 10) \rightarrow \text{deny}$ and merge it into the decision tree. This is shown in Figure \ref{fig:step2}.

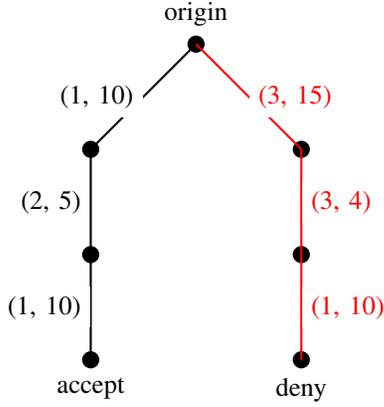
\begin{figure}[h!]
\begin{center} 
\begin{tikzpicture}[scale=.7,colorstyle/.style={circle, draw=black!100,fill=black!100, thick, inner sep=0pt, minimum size=2 mm}]
    \node at (-5,-1)[colorstyle,label=above:$\text{origin}$]{};
    \node at (-7,-3)[colorstyle]{};
    \node at (-7,-5)[colorstyle]{};
    \node at (-7,-7)[colorstyle,label=below:$\text{accept}$]{};
    \node at (-3,-3)[colorstyle]{};
    \node at (-3,-5)[colorstyle]{};
    \node at (-3,-7)[colorstyle,label=below:$\text{deny}$]{};
    
    \draw [thick](-5, -1) -- (-7, -3) node [midway, left, fill=white] {(1, 10)};
    \draw [thick](-7, -3) -- (-7, -5) node [midway, left, fill=white] {(2, 5)};
    \draw [thick](-7, -5) -- (-7, -7) node [midway, left, fill=white] {(1, 10)};
    \draw [thick, red](-5, -1) -- (-3, -3) node [midway, right, fill=white] {(3, 15)};
    \draw [thick, red](-3, -3) -- (-3, -5) node [midway, right, fill=white] {(3, 4)};
    \draw [thick, red](-3, -5) -- (-3, -7) node [midway, right, fill=white] {(1, 10)};
\end{tikzpicture}
\caption{The intermediate decision tree created from two rules connected at the origin. Observe that the outgoing edges from the origin are marked with overlapping intervals. This will be fixed in later iterations of the $addrule$ algorithm. The red edges and labels correspond to packets that rule $r$ refers to but the tree $T$ does not. }
\label{fig:step2}

\end{center}
\end{figure}

Since the intervals $(1,10)$, $(3,15)$ intersect in $(3,10)$, we split the intervals up into the three intervals $(1,2)$, $(3,10)$, and $(11,15)$. As rule 1 is the only rule that corresponds to the interval $(1,2)$, we append $r(2,:)$ to the end of the edge marked $(1,2)$. This is shown in Figure \ref{fig:step3}.

\begin{figure}[h!]

\begin{center} 
\begin{tikzpicture}[scale=.7,colorstyle/.style={circle, draw=black!100,fill=black!100, thick, inner sep=0pt, minimum size=2 mm}]
    \node at (-5,-1)[colorstyle,label=above:$\text{origin}$]{};
    \node at (-8,-3)[colorstyle]{};
    \node at (-8,-5)[colorstyle]{};
    \node at (-8,-7)[colorstyle,label=below:$\text{accept}$]{};
    \node at (-2,-3)[colorstyle]{};
    \node at (-2,-5)[colorstyle]{};
    \node at (-2,-7)[colorstyle,label=below:$\text{deny}$]{};
    \node at (-5,-3)[colorstyle, label=right:$T_1$]{};
    \node at (-6,-5)[colorstyle]{};
    \node at (-6,-7)[colorstyle,label=below:$\text{accept}$]{};
    \node at (-4,-5)[colorstyle]{};
    \node at (-4,-7)[colorstyle,label=below:$\text{deny}$]{};
    \draw [thick](-5, -1) -- (-8, -3) node [midway, left, fill=white] {(1, 2)};
    \draw [thick](-8, -3) -- (-8, -5) node [midway, left, fill=white] {(2, 5)};
    \draw [thick](-8, -5) -- (-8, -7) node [midway, left, fill=white] {(1, 10)};
    \draw [thick](-5, -1) -- (-5, -3) node [midway, fill=white] {(3, 10)};
    \draw [thick, red](-5, -3) -- (-4, -5) node [midway, right, fill=white] {(3, 4)};
    \draw [thick, red](-4, -5) -- (-4, -7) node [midway, right, fill=white] {(1, 10)};
    \draw [thick](-5, -3) -- (-6, -5) node [midway, left, fill=white] {(2, 5)};
    \draw [thick](-6, -5) -- (-6, -7) node [midway, left, fill=white] {(1, 10)};
    \draw [thick, red](-5, -1) -- (-2, -3) node [midway, right, fill=white] {(11, 15)};
    \draw [thick, red](-2, -3) -- (-2, -5) node [midway, right, fill=white] {(3, 4)};
    \draw [thick, red](-2, -5) -- (-2, -7) node [midway, right, fill=white] {(1, 10)};
\end{tikzpicture}
\caption{The first recursion: We add the parts of $I_r$ that do not intersect $I_1$ as a separate branch. We now have the subtree $T_1$ which will be resolved in the next iteration of the $addrule$ algorithm.}
\end{center}
\end{figure}
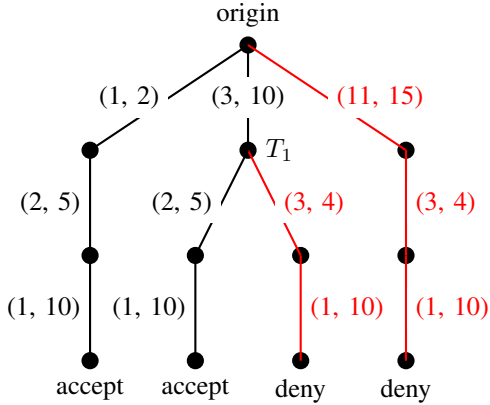

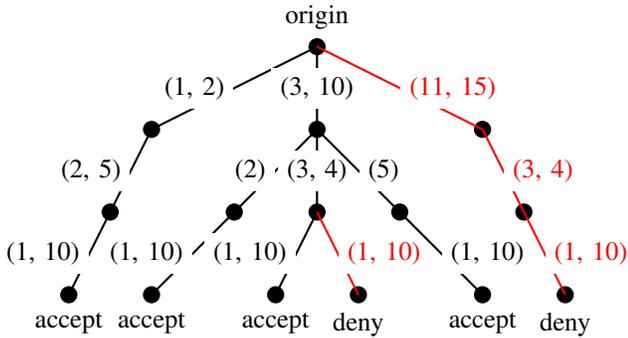
\begin{figure}[h!]

\begin{center} 
\begin{tikzpicture}[scale=.55,colorstyle/.style={circle, draw=black!100,fill=black!100, thick, inner sep=0pt, minimum size=2 mm}]
    \node at (-5,-1)[colorstyle,label=above:$\text{origin}$]{};
    \node at (-9,-3)[colorstyle]{};
    \node at (-5,-3)[colorstyle]{};
    \node at (-1,-3)[colorstyle]{};
    \node at (-10,-5)[colorstyle]{};
    \node at (-7,-5)[colorstyle]{};
    \node at (-5,-5)[colorstyle]{};
    \node at (-3,-5)[colorstyle]{};
    \node at (0,-5)[colorstyle]{};
    \node at (-11,-7)[colorstyle,label=below:$\text{accept}$]{};
    \node at (-9,-7)[colorstyle,label=below:$\text{accept}$]{};
    \node at (-6,-7)[colorstyle,label=below:$\text{accept}$]{};
    \node at (-4,-7)[colorstyle,label=below:$\text{deny}$]{};
    \node at (-1,-7)[colorstyle,label=below:$\text{accept}$]{};
    \node at (1,-7)[colorstyle,label=below:$\text{deny}$]{};

    \draw [thick](-5, -1) -- (-9, -3) node [midway, left, fill=white] {(1, 2)};
    \draw [thick](-9, -3) -- (-10, -5) node [midway, left, fill=white] {(2, 5)};
    \draw [thick](-10, -5) -- (-11, -7) node [midway, left, fill=white] {(1, 10)};
    
    \draw [thick](-5, -1) -- (-5, -3) node [midway, fill=white] {(3, 10)};
    \draw [thick](-5, -3) -- (-7, -5) node [midway, left, fill=white] {(2)};
    \draw [thick](-7, -5) -- (-9, -7) node [midway, left, fill=white] {(1, 10)};
    
    \draw [thick](-5, -3) -- (-5, -5) node [midway, fill=white] {(3, 4)};
    \draw [thick](-5, -5) -- (-6, -7) node [midway, left, fill=white] {(1, 10)};
    
    \draw [thick, red](-5, -5) -- (-4, -7) node [midway, right, fill=white] {(1, 10)};
    
    \draw [thick](-5, -3) -- (-3, -5) node [midway, right, fill=white] {(5)};
    \draw [thick](-3, -5) -- (-1, -7) node [midway, right, fill=white] {(1, 10)};
    
    \draw [thick, red](-5, -1) -- (-1, -3) node [midway, right, fill=white] {(11, 15)};
    \draw [thick, red](-1, -3) -- (0, -5) node [midway, right, fill=white] {(3, 4)};
    \draw [thick, red](0, -5) -- (1, -7) node [midway, right, fill=white] {(1, 10)};
    
\end{tikzpicture}
\caption{The second iteration of the $addrule$ algorithm. All the edges of depth two are marked with disjoint intervals.}
\label{fig:step4}

\end{center}
\end{figure}

Now that we have completed the decision tree from the origin, we need to merge rule $r_2(2,:)$, into $T_1$, the subtree whose root node is the endpoint of the edge labelled $(3,10)$. So we take the first interval in rule $r_2(2,:)$, (3,4), and the intervals in the outgoing edges of the root node of $T_1$, (2,5). These intervals intersect at $(3,4)$, and rule 2 is not present at $(2)$, and $(5)$. This is shown in Figure \ref{fig:step4}.

So now, our final step is to merge the rule $(1,10)\rightarrow deny$, with the tree, that consists of a single edge $(1,10)\rightarrow accept$. Since the interval of the tree's outgoing edge, is a superset of the rule's interval, the rule is completely subsumed by the tree and we just have the tree $(1,10)\rightarrow accept$. 

We have now reached the base case. We have a tree that just consists of the action accept and a rule that just consists of the action deny. As specified in the algorithm, the accept action takes highest priority and so we have our decision tree. 

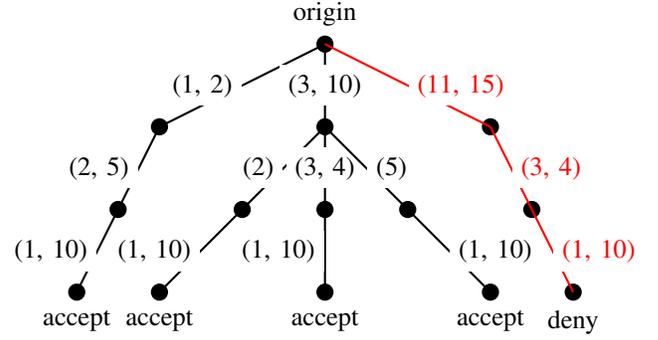
\begin{figure}[h!]
\begin{center} 
\begin{tikzpicture}[scale=.55,colorstyle/.style={circle, draw=black!100,fill=black!100, thick, inner sep=0pt, minimum size=2 mm}]
    \node at (-5,-1)[colorstyle,label=above:$\text{origin}$]{};
    \node at (-9,-3)[colorstyle]{};
    \node at (-5,-3)[colorstyle]{};
    \node at (-1,-3)[colorstyle]{};
    \node at (-10,-5)[colorstyle]{};
    \node at (-7,-5)[colorstyle]{};
    \node at (-5,-5)[colorstyle]{};
    \node at (-3,-5)[colorstyle]{};
    \node at (0,-5)[colorstyle]{};
    \node at (-11,-7)[colorstyle,label=below:$\text{accept}$]{};
    \node at (-9,-7)[colorstyle,label=below:$\text{accept}$]{};
    \node at (-5,-7)[colorstyle,label=below:$\text{accept}$]{};
    \node at (-1,-7)[colorstyle,label=below:$\text{accept}$]{};
    \node at (1,-7)[colorstyle,label=below:$\text{deny}$]{};

    \draw [thick](-5, -1) -- (-9, -3) node [midway, left, fill=white] {(1, 2)};
    \draw [thick](-9, -3) -- (-10, -5) node [midway, left, fill=white] {(2, 5)};
    \draw [thick](-10, -5) -- (-11, -7) node [midway, left, fill=white] {(1, 10)};
    
    \draw [thick](-5, -1) -- (-5, -3) node [midway, fill=white] {(3, 10)};
    \draw [thick](-5, -3) -- (-7, -5) node [midway, left, fill=white] {(2)};
    \draw [thick](-7, -5) -- (-9, -7) node [midway, left, fill=white] {(1, 10)};
    
    \draw [thick](-5, -3) -- (-5, -5) node [midway, fill=white] {(3, 4)};
    \draw [thick](-5, -5) -- (-5, -7) node [midway, left, fill=white] {(1, 10)};
    
    \draw [thick](-5, -3) -- (-3, -5) node [midway, right, fill=white] {(5)};
    \draw [thick](-3, -5) -- (-1, -7) node [midway, right, fill=white] {(1, 10)};
    
    \draw [thick, red](-5, -1) -- (-1, -3) node [midway, right, fill=white] {(11, 15)};
    \draw [thick, red](-1, -3) -- (0, -5) node [midway, right, fill=white] {(3, 4)};
    \draw [thick, red](0, -5) -- (1, -7) node [midway, right, fill=white] {(1, 10)};
    
\end{tikzpicture}
\caption{The third and final iteration of the $addrule$ algorithm. We now have a firewall decision tree whose edge labels are disjoint. }
\label{merged tree}

\end{center}
\end{figure}

We now have a decision tree representation of the rules $ (1, 10)(2, 5)(1, 10) \rightarrow \text{accept}$ and $(3, 15)(3, 4)(1, 10) \rightarrow \text{deny}$. We can translate the decision tree back into a rule list by taking the paths from the origin to every leaf node. So the rules $ (1, 10)(2, 5)(1, 10) \rightarrow \text{accept}$ and $(3, 15)(3, 4)(1, 10) \rightarrow \text{deny}$ are equivalent to the list of rules

\begin{align*} 
(1,2)(2,5)(1,10)&\rightarrow \text{accept}\\
(3,10)(2)(1,10)&\rightarrow \text{accept}\\
(3,10)(3,4)(1,10)&\rightarrow \text{accept}\\
(3,10)(5)(1,10)&\rightarrow \text{accept}\\
(11,15)(3,4)(1,10)&\rightarrow \text{deny}
\end{align*}

Since the predicates of each rule refers to disjoint sets, the order of the rules is irrelevant. 

Using this algorithm to translate a firewall into a decision tree, the tree can grow quite quickly. We now introduce the final step of this algorithm which groups adjacent rules together to create a reduced tree. We say that two nodes $n_1$ and $n_2$ of a firewall decision tree are \textit{adjacent} if and only if the following three conditions hold:
\begin{enumerate}
    \item $n_1$ and $n_2$ are sibling nodes \textit{i.e.}, they have the same parent node.
    \item Let $e_1$ and $e_2$ denote the edges that connect $n_1$ and $n_2$ to their mutual parent node. $e_1$ and $e_2$ are marked with two intervals $(a,b)$ and $(c, d)$, where $c = b + 1$.
    \item The sub-trees, with $n_1$ and $n_2$ as their respective roots are identical. 
\end{enumerate}

Intuitively, we can imagine each rule in the firewall as being a rectangular hyperprism. In our examples, $d = 3$ and we can imagine each rule as a rectangular prism. When two rules are adjacent we can represent them both using a single rule. This is the same as putting the two rectangular prisms together to form a single rectangular prism. The three conditions is a mathematically rigorous way of stating this.

Two adjacent nodes can be grouped together as follows
\begin{itemize}
    \item Starting at the terminal nodes, group together all adjacent nodes. 
    \item For every node at depth $i$ for $i\in {d, d-1, \dots, 2, 1}$, group together all adjacent nodes. 
\end{itemize}
By starting from the bottom of the tree and proceeding upwards, we ensure that the sub-trees below two adjacent nodes are actually identical and not just equivalent. In this algorithm, each node of the tree is visited exactly once, and thus the algorithm runs in linear time with respect to the number of nodes. 

Consider the decision tree in Figure \ref{merged tree}. After grouping adjacent nodes together, we obtain the tree. 

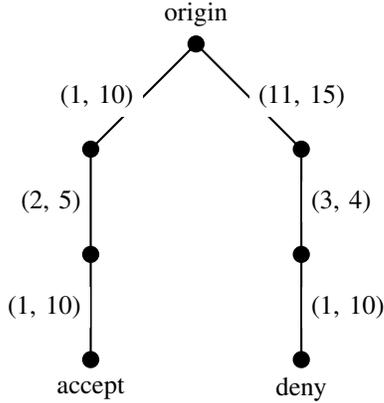
\begin{figure}[h!]
\begin{center} 
\begin{tikzpicture}[scale=.7,colorstyle/.style={circle, draw=black!100,fill=black!100, thick, inner sep=0pt, minimum size=2 mm}]
    \node at (-5,-1)[colorstyle,label=above:$\text{origin}$]{};
    \node at (-7,-3)[colorstyle]{};
    \node at (-3,-3)[colorstyle]{};
    \node at (-7,-5)[colorstyle]{};
    
    \node at (-3,-5)[colorstyle]{};
    \node at (-7,-7)[colorstyle,label=below:$\text{accept}$]{};

    \node at (-3,-7)[colorstyle,label=below:$\text{deny}$]{};
    \draw [thick](-5, -1) -- (-7, -3) node [midway, left, fill=white] {(1, 10)};
    \draw [thick](-7, -3) -- (-7, -5) node [midway, left, fill=white] {(2, 5)};
    \draw [thick](-7, -5) -- (-7, -7) node [midway, left, fill=white] {(1, 10)};
    
    \draw [thick](-5, -1) -- (-3, -3) node [midway, right, fill=white] {(11, 15)};
    \draw [thick](-3, -3) -- (-3, -5) node [midway, right, fill=white] {(3, 4)};
    \draw [thick](-3, -5) -- (-3, -7) node [midway, right, fill=white] {(1, 10)};
    
\end{tikzpicture}
\caption{The decision trees with adjacent rules grouped together. }
\label{grouped tree}

\end{center}
\end{figure}

\subsection{Whitelist and blacklist translation}
It is a simple matter to extract the whitelist and blacklist translations of a firewall policy from a decision tree. In this subsection we present an algorithm that does this. A firewall decision tree is called \textit{complete} when every packet is mapped to an action. We can make each firewall decision tree complete by merging it with a \textit{deny all} or an \textit{accept all} rule. We continue to use the example in Figure \ref{grouped tree}.

We merge this tree with a deny all rule to obtain the following decision tree. For simplicity, we assume that each packet header field takes on a number from 0 to 15. 

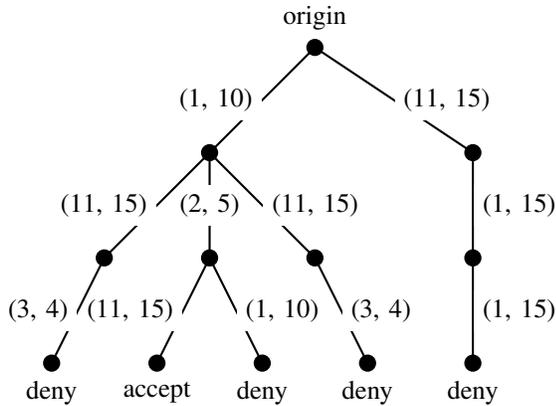
\begin{figure}[h!]
\begin{center} 
\begin{tikzpicture}[scale=.7,colorstyle/.style={circle, draw=black!100,fill=black!100, thick, inner sep=0pt, minimum size=2 mm}]
    \node at (-5,-1)[colorstyle,label=above:$\text{origin}$]{};
    \node at (-7,-3)[colorstyle]{};
    \node at (-2,-3)[colorstyle]{};
    \node at (-7,-5)[colorstyle]{};
    \node at (-9, -5)[colorstyle]{};
    \node at (-5, -5)[colorstyle]{};
    \node at (-2,-5)[colorstyle]{};
    
    \node at (-10,-7)[colorstyle,label=below:$\text{deny}$]{};
    \node at (-8,-7)[colorstyle,label=below:$\text{accept}$]{};
    \node at (-6,-7)[colorstyle,label=below:$\text{deny}$]{};
    \node at (-4,-7)[colorstyle,label=below:$\text{deny}$]{};
    \node at (-2,-7)[colorstyle,label=below:$\text{deny}$]{};
    
    \draw [thick](-5, -1) -- (-2, -3) node [midway, right, fill=white] {(11, 15)};
    \draw [thick](-2, -3) -- (-2, -5) node [midway, right, fill=white] {(1, 15)};
    \draw [thick](-2, -5) -- (-2, -7) node [midway, right, fill=white] {(1, 15)};
    \draw [thick](-5, -1) -- (-7, -3) node [midway, left, fill=white] {(1, 10)};
    \draw [thick](-7, -3) -- (-7, -5) node [midway, fill=white] {(2, 5)};
    \draw [thick](-7, -5) -- (-6, -7) node [midway, right, fill=white] {(1, 10)};
    \draw [thick](-7, -5) -- (-8, -7) node [midway, left, fill=white] {(11, 15)};
    \draw [thick](-7, -3) -- (-5, -5) node [midway, right, fill=white] {(11, 15)};
    \draw [thick](-5, -5) -- (-4, -7) node [midway, right, fill=white] {(3, 4)};
    \draw [thick](-7, -3) -- (-9, -5) node [midway, left, fill=white] {(11, 15)};
    \draw [thick](-9, -5) -- (-10, -7) node [midway, left, fill=white] {(3, 4)};
    
\end{tikzpicture}
\caption{A complete decision tree given by taking the decision tree in Figure \ref{grouped tree} and adding a deny rule using the addrule algorithm. Now every possible packet in this space is assigned a unique path from the tree's origin to a leaf node.}
\label{completetree}
\end{center}
\end{figure}

We can translate the firewall policy into a whitelist by going through the tree and taking a list of all the paths from the origin to the decision nodes marked ``accept". Because the paths in the decision tree represent disjoint sets of packages, it does not matter what order we take these paths. The whitelist corresponding to Figure \ref{completetree} is given by:

$$
(1, 10)(2, 5)(1, 10) \rightarrow \text{accept}
$$
$$
(1, 15)(1, 15)(1, 15) \rightarrow \text{deny},
$$
and the corresponding blacklist is given by 
$$
(1, 10)(1)(1, 15) \rightarrow \text{deny}
$$
$$
(1, 10)(2,5)(11, 15) \rightarrow \text{deny}
$$
$$
(1, 10)(6, 15)(1, 15) \rightarrow \text{deny}
$$
$$
(11, 15)(1, 15)(1, 15) \rightarrow \text{deny}
$$
$$
(1, 15)(1, 15)(1, 15) \rightarrow \text{accept}.
$$

In summary, we have shown how to create a Boolean decision tree from a rule list using the addrule algorithm. We've seen how it takes a rule and recursively finds the components of the rule disjoint form the tree and adds these components to the tree. 

\section{Extracting Boolean expressions from firewall decision trees}

In this section, we provide an algorithm that extracts Boolean expressions from firewall decision trees. We then show that based on this algorithm, the size of the Boolean expressions are within a polynomial factor of each other. 

Every firewall rule list consists of rules of the form \newline$<predicate>\rightarrow <action>$, We will first show how each predicate can be expressed as a Boolean formula and how, these formulas can be combined together to give a Boolean expression that describes the firewall policy. Consider the following rule:
$$
x_1 \in [0,108] \wedge x_2 \in [21,655] \wedge x_3 \in [7,616] \rightarrow accept.
$$ 
Each rule consists of $n$ statements of the form $x_i \in [a,b]$, where $a$ and $b$ are (not necessarily distinct) integers. We show how to derive an expression for $x_i\in [a,b]$ by using an interval tree. We then use the interval tree to augment the firewall decision tree to obtain a Boolean decision diagram. Using some simple combinatorial arguments, we show that the size of the firewall policy, represented as a Boolean expression in CNF/DNF can be bounded above by the size of the firewall policy expressed as a blacklist/whitelist multiplied by a constant factor. 

The key technique used is that numbers can be represented as bit vectors. For example, an address segment is a number between 0 and 255. At a lower level, the address segment is a vector of 8 bits. The 8-bit number $x$ can be represented by the bit vector $(x_7, . . . , x_0)$, where each $x_i$ is a Boolean variable. Using $f$ to represent 0 and using $t$ to represent 1, the condition that the bit vector $x$ is equal to 5, is $x_7, . . . , x_0 = f,f,f,f,f,t, f, t$. This yields the Boolean expression $\neg x_7\wedge\neg x_6\wedge\neg x_5\wedge\neg x_4\wedge\neg x_3\wedge x_2\wedge \neg x_1\wedge x_0$. 

\subsection{Segment trees}
In the previous subsection, we showed how to represent numbers as bit vectors and how we can create a Boolean expression that checks equality. In this subsection, we show how to efficiently represent intervals using segment trees. The simplest method to represent intervals is to represent all the numbers in the interval as bit vectors and take their disjunction. For an interval $[a,b]$, this results in a Boolean expression that has $b-a$ clauses. 

We can reduce this to $O(4\log(B))$ clauses, where $B$ is the size of the address space, by expressing our interval using a segment tree \cite{berg1997computational}. A segment tree is a data structure designed to represent interval data. A segment for a set of intervals taking values in the range $(0, B)$ uses $O(B \log B)$ storage and can be built in $O(B \log B)$ time \cite{berg1997computational}. There are certain optimisations that can make the tree smaller by placing restrictions of the types of packets that can be encountered. For the sake of simplicity, we do not consider these optimisations. 

In Figure \ref{fig:intervaltree} we show a depth-4 segment tree. This tree can represent intervals between 0 and 15. 

\begin{figure}[h]
    \centering
    \includegraphics[width=0.5\textwidth]{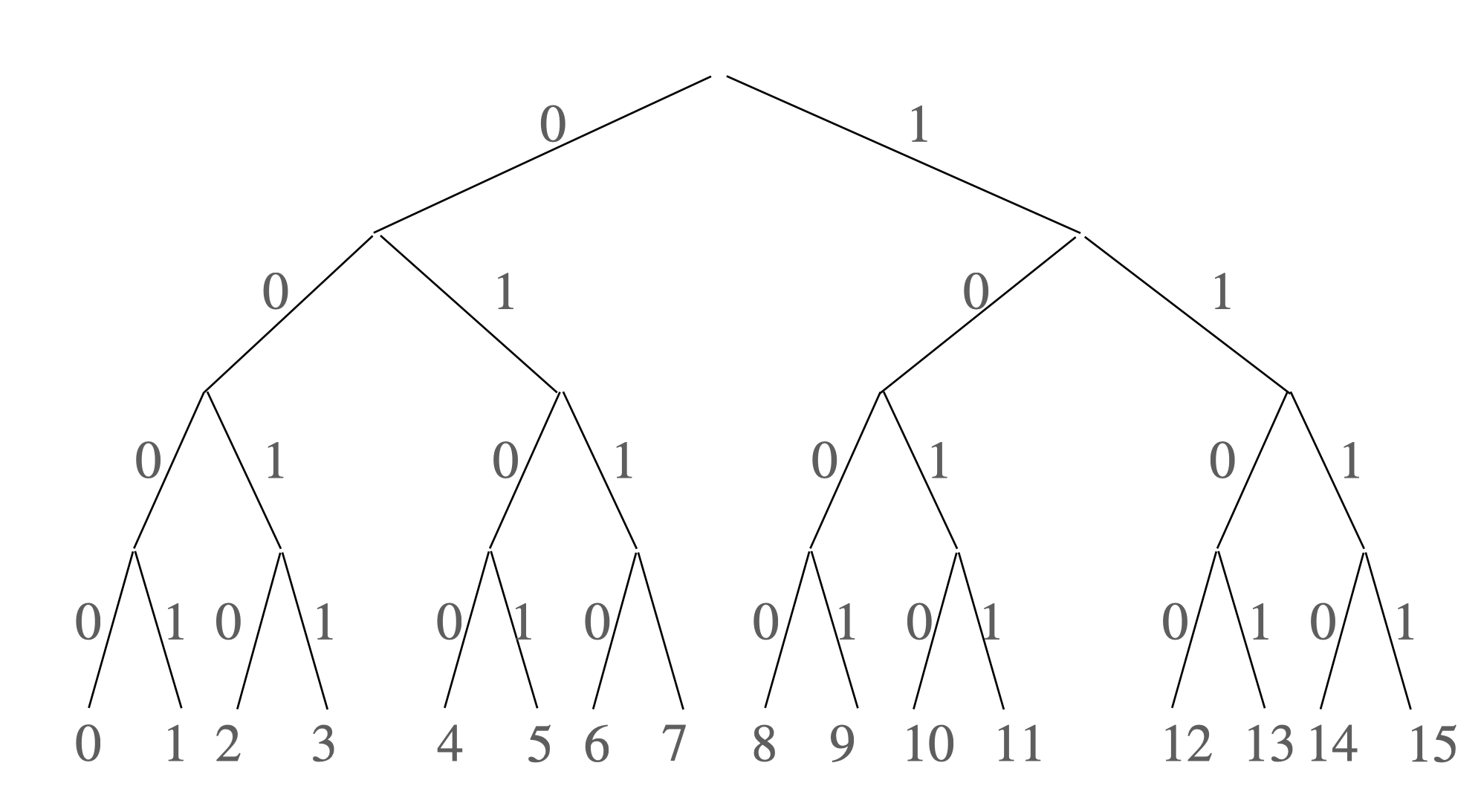}
    \caption{An interval tree of depth 4. The edge labels represent the Boolean variables used to express the numbers on the leaf nodes. }
    \label{fig:intervaltree}
\end{figure}

This is a binary tree. We assign a Boolean variable to each layer of the tree. The origin is assigned $x_0$, the second layer is assigned $x_1$, and so on. Every node in the terminal layer, can be written in binary notation, such that the binary notation of the number at the terminal node also represents the instructions for how to reach that number in the tree. For instance the number 4 can be written in binary as 0011 or be given as the set of directions: ``go left, left, right, and right down the tree". We can now identify all the nodes in the segment tree by a binary string and call this the \textit{sequence of a node}.  

Now that we have a decision tree, we give an example of expressing the interval $[3,13]$ using the tree and extracting a Boolean expression that is only satisfied by a bit vector with value between 3 and 13 inclusive. We do this by first identifying all the paths in the tree that terminate at one of these numbers. Then we alter the decision tree according to the following rules:

\begin{figure}[h]
    \centering
    \includegraphics[width=0.5\textwidth]{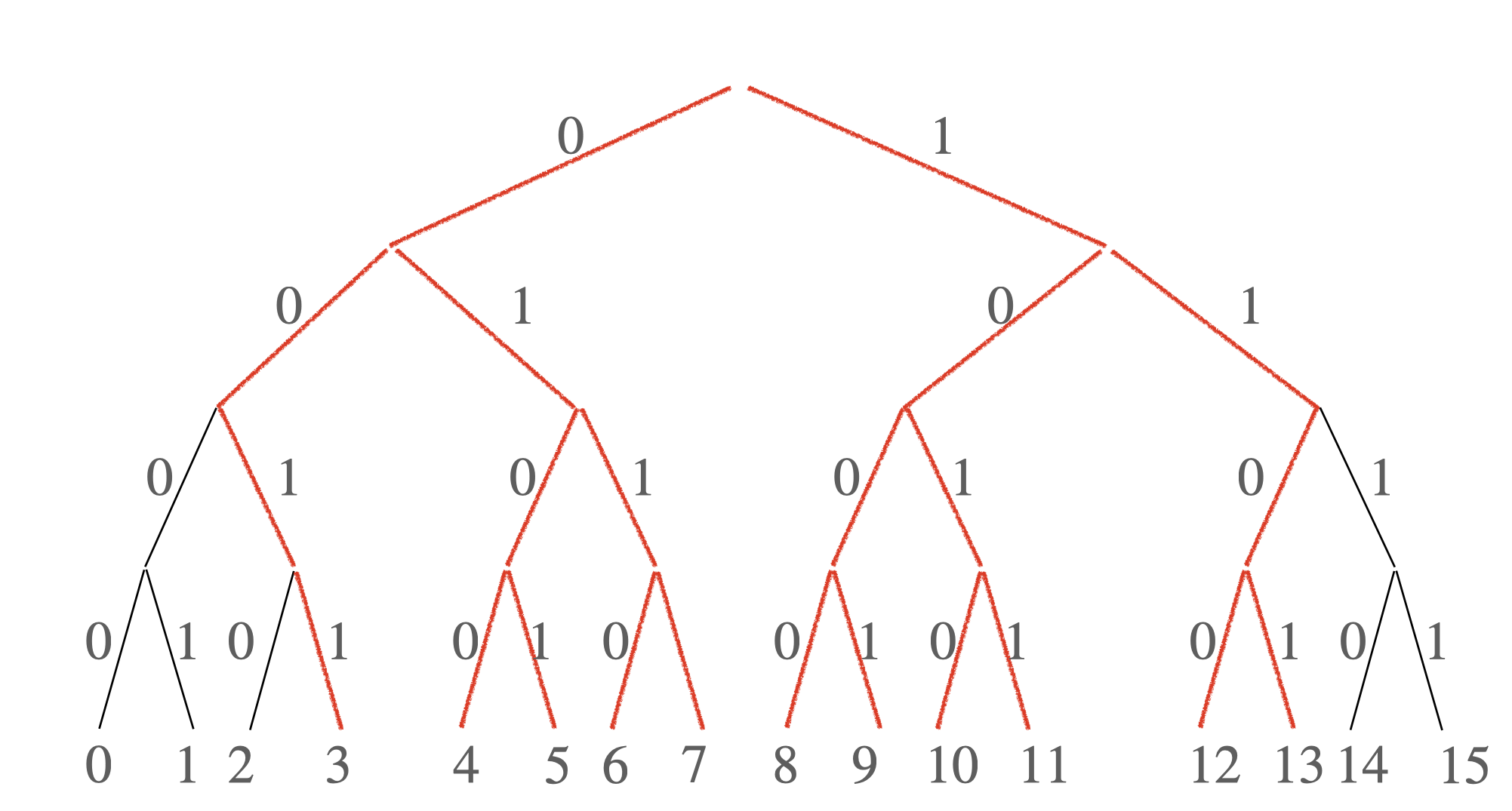}
    \caption{The interval tree from Figure \ref{fig:intervaltree}, with the paths from the root to nodes whose labels are between 3 and 13 shown in red. }
    \label{fig:intervaltree2}
\end{figure}

\begin{enumerate}
    \item If a node $n$ and it's neighbour $n'$ are marked and both have no marked descendants, then remove the mark from both $n$ and $n'$ and mark the parent node of $n$ and $n'$. 
    \item Terminate this process when it is impossible to remove any more marked nodes. 
\end{enumerate}

After applying these rules to the decision tree, we call the result a \textit{reduced decision tree}. We use this terminology because the Boolean expression associated with the interval tree is reduced. We apply these rules to Figure \ref{fig:intervaltree2}, we obtain the reduced, interval tree.

\begin{figure}[h]
    \centering
    \includegraphics[width=0.5\textwidth]{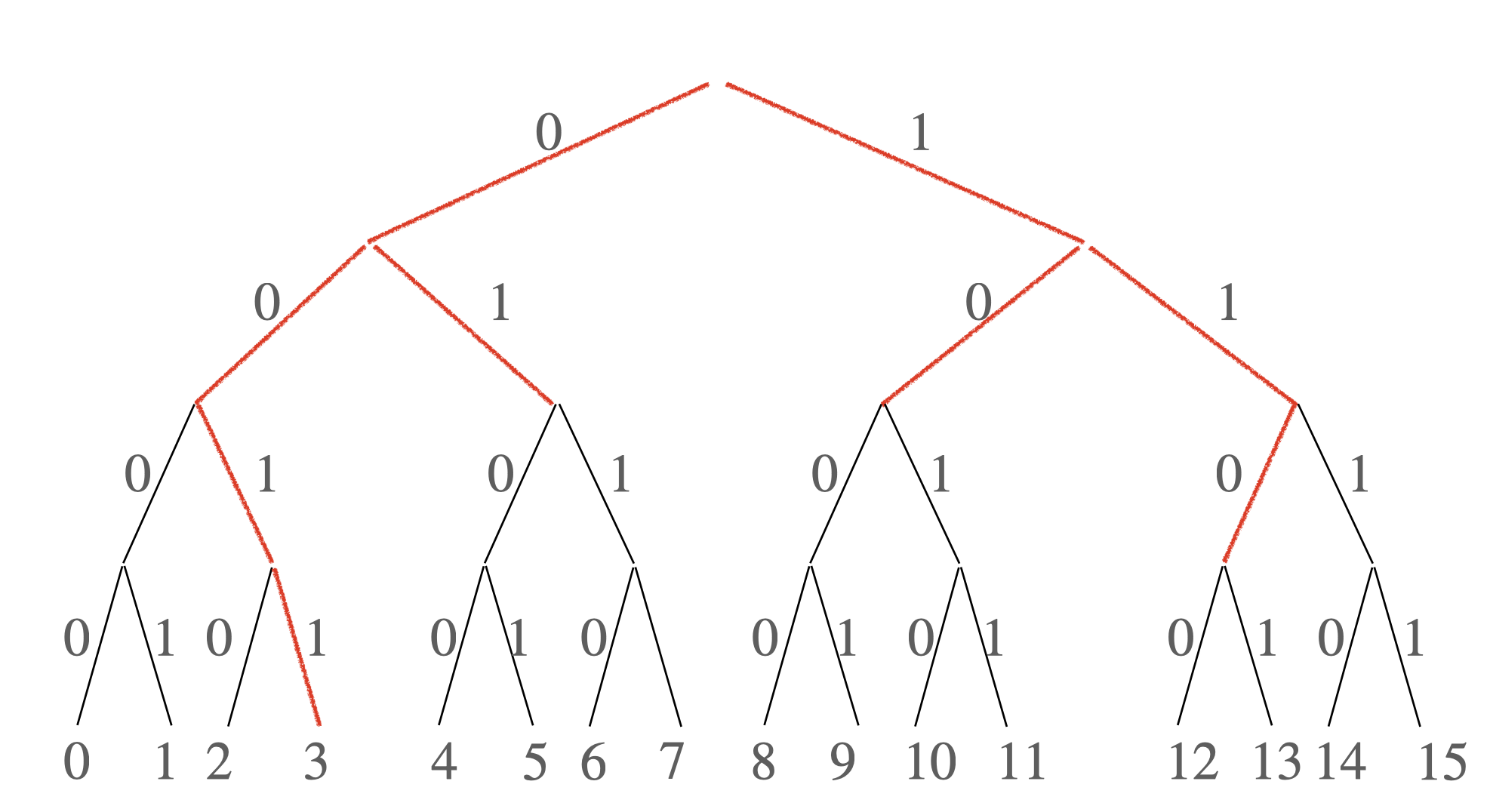}
    \caption{The interval tree from figure \ref{fig:intervaltree2} with the redundant edges removed.  }
    \label{fig:intervaltree3}
\end{figure}

Now that we have the reduced interval tree, we can obtain a Boolean expression for the interval $[3,13]$, by taking the disjunction of every marked path from the root node. From our reduced interval tree in Figure \ref{fig:intervaltree3}, we obtain the Boolean expression
$$
\neg x_0\neg x_1 x_2 x_3 \vee \neg x_0 x_1 \vee x_0 \neg x_1 \vee x_0 x_1 \neg x_2.
$$

\subsection{Boolean decision trees of firewall policies}
We now show how to combine the segment tree data structures and the firewall decision trees to create a Boolean decision tree representation of a firewall policy. We start by taking the firewall decision tree from Section \ref{Firewall decision trees} and represent all the outgoing edges in the firewall decision tree using a segment tree. Here we show an example involving the complete decision tree from Figure \ref{completetree}. For simplicity, we denote the subtrees rooted at the end of the $(1, 10)$ edge as $T_1$ and $T_2$, respectively.

The main idea is to take the intervals in each edge of the decision tree and to replace this edge with an interval tree. 

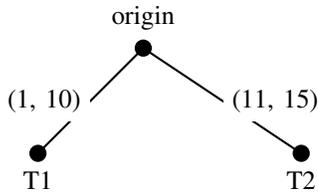
\begin{figure}[h!]
\begin{center} 
\begin{tikzpicture}[scale=.7,colorstyle/.style={circle, draw=black!100,fill=black!100, thick, inner sep=0pt, minimum size=2 mm}]
    \node at (-5,-1)[colorstyle,label=above:$\text{origin}$]{};
    \node at (-7,-3)[colorstyle,label=below:$\text{T1}$]{};
    \node at (-2,-3)[colorstyle,label=below:$\text{T2}$]{};

    \draw [thick](-5, -1) -- (-2, -3) node [midway, right, fill=white] {(11, 15)};
    \draw [thick](-5, -1) -- (-7, -3) node [midway, left, fill=white] {(1, 10)};
    
\end{tikzpicture}
\caption{The decision tree whose edges are marked with intervals}
\label{completetree}

\end{center}
\end{figure}

We can express the two intervals $(1, 10)$ and $(11, 16)$ using the segment tree in Figure \ref{completetree}. 

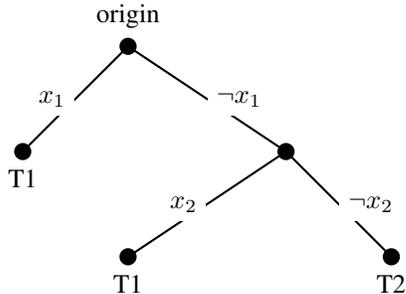
\begin{figure}[h!]
\begin{center} 
\begin{tikzpicture}[scale=.7,colorstyle/.style={circle, draw=black!100,fill=black!100, thick, inner sep=0pt, minimum size=2 mm}]
    \node at (-5,-1)[colorstyle,label=above:$\text{origin}$]{};
    \node at (-7,-3)[colorstyle,label=below:$\text{T1}$]{};
    \node at (-2,-3)[colorstyle]{};
    \node at (-5, -5)[colorstyle,label=below:$\text{T1}$]{};
    \node at (0, -5)[colorstyle,label=below:$\text{T2}$]{};

    \draw [thick](-5, -1) -- (-7, -3) node [midway, left, fill=white] {$x_1$};
    \draw [thick](-5, -1) -- (-2, -3) node [midway, right, fill=white] {$\neg x_1$};
    \draw [thick](-2, -3) -- (-5, -5) node [midway, left, fill=white] {$ x_2$};
    \draw [thick](-2, -3) -- (0, -5) node [midway, right, fill=white] {$\neg x_2$};

\end{tikzpicture}
\caption{The same decision tree whose vertices are marked with Boolean variables. }
\label{completetree}

\end{center}
\end{figure}

We can proceed down the tree and express the intervals as binary decision trees until the entire firewall policy is translated into a BDD. 

\subsection{CNF and DNF reformat}
\label{CNF and DNF reformat}
\begin{theorem}
Given a firewall decision tree expressed as a Boolean decision diagram, the firewall policy expressed in CNF and in DNF can be extracted from the diagram in time that is polynomial in the number of firewall rules $n$. 
\end{theorem}

\begin{proof}
The Boolean expression in DNF can be found by going through the Boolean decision diagram and counting the number of paths from the root node to the terminal node marked ``1". 

The number of variables in each clause is at most $\log(B)$, where $B$ is the largest possible interval that we can have for $[a,b]$. The number of clauses is at most $2\log(B)$. There are at most $d$ intervals in the rule. 

Here we provide a loose upper bound on the size of the number of paths from the root node to the terminal node marked ``1". In the firewall decision tree for a policy $P$, there are at most $(2n - 1)^d$ terminal nodes. This means that there are at most $(2n - 1)^d$ rules in the firewall expressed as a whitelist. We consider the interval tree for each interval used to mark the edges in the firewall decision tree. There are at most $(2\log(B))$, clauses in the Boolean expression that describe each interval. We can build a DNF representation of the firewall policy by going along every path from the root node to the decision node marked ``1" and for every edge crossed, adding the path from the interval tree to our clause and then taking the disjunction of all of the clauses. The number of clauses in a Boolean expression is at most $(2\log(B))^d(2n-1)^d$, which remains polynomial in the number of rules $n$. 
\end{proof}

We can obtain a CNF representation of the firewall policy $P$ by counting all the paths from the root node to the terminal node labelled with "0". This is the negation of $P$ expressed as a Boolean expression in DNF, denoted by $\neg P$. We then take the negation of $\neg P$ to obtain $\neg \neg P$, an expression equivalent to $P$. Since $\neg P$ is in DNF, we can obtain an expression for $P$ in CNF through a straightforward application of DeMorgan's law
$$
\neg \bigwedge_{i = 1}^{C}\big( \bigvee_{j = 1}^{K_i} t_i^j\big)
=\bigvee_{i = 1}^{C}\big( \bigwedge_{j = 1}^{K_i} \neg t_i^j\big).
$$
Since the number of nodes in the firewall decision tree is upper-bounded by $(2n-1)^d$, the number of rules in a firewall expressed as a blacklist is also $(2n-1)^d$ and therefor the number of clauses in the CNF of the Boolean expression is at most $(2\log(B))^d(2n-1)^d$.

\begin{corollary}
Let $P$ be a family of firewall policies, let $P_{rule}$ be the size of these policies expressed as a rule list, and let $P_{FO}$ be the size of the firewall policy expressed in first order logic. There exists some family of policies $P$ such that $P_{rule} \in O(2^{P_{FO}})$.
\end{corollary}

\begin{proof}
In section \ref{CNF and DNF reformat} saw that the DNF and CNF of firewall policies are within a degree $d$ polynomial factor of one another. If we were to consider packet spaces with $d+1$ fields then, we would obtain a family of Boolean expressions, whose CNF and DNF forms are within a degree $d+1$ polynomial of each other. It is well known that Boolean formulas in DNF can be exponentially larger than the expression in CNF \cite{Velev}. Let $P$ be a sequence of Boolean expressions that have this property. 

Since the Boolean expression in DNF is within a polynomial factor of the firewall policy as a rule list, we need an exponential number of rules in the rule list to create these policies. So the size of these policies expressed as rule-lists must be exponentially larger than the formula expressed in first-order logic in CNF. 
\end{proof}

\section{Conclusion}
In this paper we have made and have justified the following points. 
Under the assumption of finitely many packet header fields, Firewall policies can be expressed as Boolean expressions such that the expression in CNF and the expression in DNF have sizes that are both $d$ dimensional polynomials of the original number of rules. Firewall policies in CNF can be translated to their DNF counterparts in polynomial time, an exponential improvement over naive expansion. 

In the future we hope to extend these 

\section*{Acknowledgment}

The authors would like to thank Boeing Defence Australia, and the ARC Centre of Excellence for Mathematical and Statistical Frontiers.

\ifCLASSOPTIONcaptionsoff
  \newpage
\fi



%

\bibliographystyle{acm}
\bibliography{ref}

%

\begin{IEEEbiography}[{\includegraphics[width=1in,height=1.25in,clip,keepaspectratio]{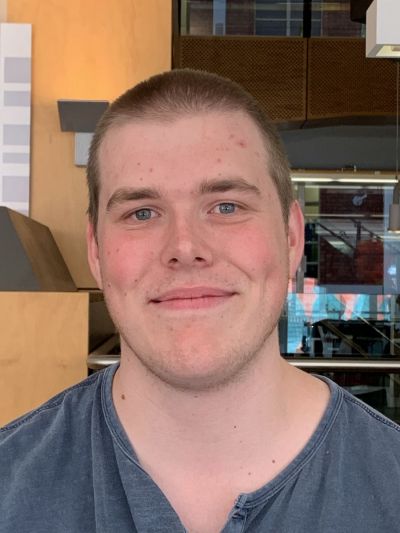}}]{Adam Hamilton}
obtained his bachelor's in mathematical science and his honours degree in mathematical and computer science from the University of Adelaide. He is currently pursuing a doctorate in mathematical and computer science at the University of Adelaide under the supervision of Professor Matthew Roughan and Dr Giang Nguyen. His research interests include the theory of computation, computational complexity, randomised algorithms, and time difference of arrival geolocation.  
\end{IEEEbiography}

\begin{IEEEbiography}[{\includegraphics[width=1in,height=1.25in,clip,keepaspectratio]{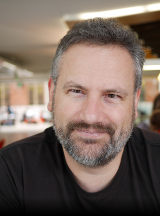}}]{Prof Matthew Roughan} obtained his PhD in Applied Mathematics from the
University of Adelaide in 1994. He has since worked in various
industry related positions with Australian Defence, with Ericsson at
the the University of Melbourne and at AT\&T Shannon Research Labs in
the United States. For the last decade or so he has worked in the School of
Mathematical Sciences at the University of Adelaide, in South
Australia. His research interests range from stochastic modelling to
measurement and management of networks like the Internet.  In 2018 he
was elected to be a Fellow of the ACM, and in 2019 a Fellow of the
IEEE for his work on Internet measurement, and he was a Chief
Investigator in the ARC Centre of Excellence for Mathematical and
Statistical Frontiers. 
\end{IEEEbiography}


\begin{IEEEbiography}[{\includegraphics[width=1in,height=1.25in,clip,keepaspectratio]{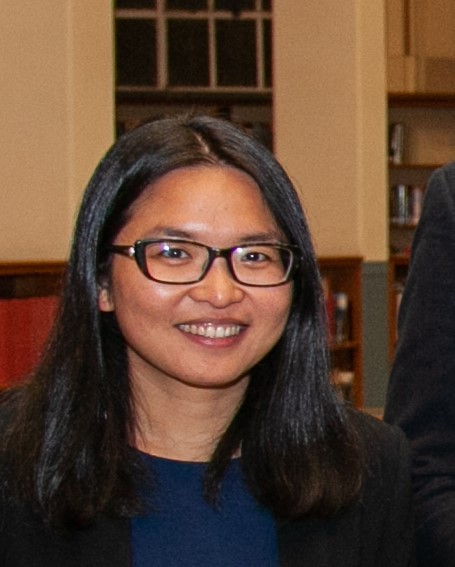}}]{Dr Giang Nguyen} Giang Nguyen is a Senior Lecturer in Applied Mathematics at the University of Adelaide. Her research interests include stochastic differential equations, regime-switching diffusions, matrix-analytic methods, branching processes, and the Hamiltonian cycle problem. She received a PhD from the University of South Australia in 2009, and completed her postdoctoral studies at the Universite libre de Bruxelles (2009-2012) and at the University of Adelaide (2012-2013). She was the inaugural Treasurer of the Australian Mathematical Society Women in Mathematics Special Interest Group (WIMSIG) (2013-2018), and the Director of Gender Equity, Diversity and Inclusion of the Faculty of Engineering, Computer and Mathematical Sciences, the University of Adelaide (2019-2020). She was a 2019 South Australia Tall Poppy Winner. 
\end{IEEEbiography}




\end{document}